\documentclass[1p,final]{elsarticle}
\usepackage{amsfonts,color,morefloats,pslatex}
\usepackage{amssymb,amsthm, amsmath,latexsym}
\usepackage{makecell}
\usepackage{graphicx}
\usepackage{float}

\allowdisplaybreaks[4]

\newcommand{\rev}[1]{\textcolor{red}{#1}}

\newtheorem{remark}{Remark}
\newtheorem{theorem}{Theorem}
\newtheorem{lemma}[theorem]{Lemma}

\newtheorem{corollary}[theorem]{Corollary}

\newtheorem{example}{Example}

\newcommand{\ord}{{\mathrm{ord}}}

\newcommand{\C}{{\mathcal{C}}}

\usepackage{lineno}

\begin{document}

\begin{frontmatter}



\title{ Non-RS type cyclic MDS codes over finite fields via cyclotomic field reduction
\tnotetext[fn1]{}}

\author[cx]{Can Xiang}
\address[cx]{College of Mathematics and Informatics, South China Agricultural University, Guangzhou, Guangdong 510642, China}
\ead{cxiangcxiang@hotmail.com}
\author[cmt]{Chunming Tang}
\address[cmt]{School of Information Science and Technology, Southwest Jiaotong University, Chengdu 610031, China}
\ead{tangchunmingmath@163.com}


\begin{abstract}
Cyclic maximum distance separable (MDS for short) codes are a special subclass of linear codes and have received a lot of attention, as these codes have very important applications in many areas including quantum codes, designs and finite geometry. However, the existing construction methods for cyclic MDS codes mainly rely on strict restrictions on certain parameters or involve relatively complex procedures. In this paper, we investigate this problem further via norm reduction in cyclotomic fields. By converting the verification of the MDS property over a finite field into checking non-zero minors in characteristic zero, we propose a construction method of cyclic MDS codes over finite fields via cyclotomic field reduction. Based on this method, we obtain several cyclic MDS codes over finite fields and many non-RS type cyclic MDS codes are produced. Compared with the existing construction method in [Non-Reed-Solomon Type Cyclic MDS Codes, IEEE Trans. Inf. Theory, 71(5):
3489--3496, 2025], we give a new and relatively direct characterization method for constructing cyclic MDS codes over finite fields and the obtained non-RS cyclic MDS codes have flexible parameters in terms of the code length $n$, the defining set $J$, and the underlying finite field $\mathbb{F}_q$.
\end{abstract}


\end{frontmatter}

\section{Introduction}\label{Sec-introduct}

Let $p$ be a prime and $q = p^m$ for some positive integer $m$. Let $\mathbb{F}_q$ be the finite field
of cardinality $q$. An $[n,\, k,\,d]$ linear code $\C$ over $\mathbb{F}_q$ is a $k$-dimensional subspace of $\mathbb{F}_q^n$ with minimum (Hamming) distance $d$. An $[n,\, k,\,d]$ linear code $\C$ is said to be {\em cyclic} if $(c_0,c_1, \cdots, c_{n-1}) \in \C$ implies $(c_{n-1}, c_0, c_1, \cdots, c_{n-2}) \in \C$ for all $(c_0,c_1, \cdots, c_{n-1}) \in \C $. An $[n,\, k,\,d]$ linear code $\C$ is said to be maximum distance separable (MDS for short) if its parameters reach the Singleton bound,  i.e., $d=n-k + 1$. It is well known that the dual of an MDS code is also an MDS code.

It is well known that an MDS code is a special linear code and has very important applications in cryptography, designs, finite geometry and so on (see, for example, \cite{ma2003,dingwang2005}). Thus the study on MDS codes has attracted much attention. Some properties of MDS codes have been well investigated and some MDS codes were obtained by various methods (see, for example, \cite{liu2021,liu2026,roth1989,kok2016,ma1993,sui2022,br1999,dsdy2013,dau2014,heng2024,chen2025,jin2026,bee2017,bee2022,Yan2025,Wang2024,Chen2018}). An MDS code is called of Reed--Solomon (RS for short) type if it is monomially equivalent to a generalized Reed--Solomon (GRS) code, and an MDS code which is not equivalent to a GRS code is called a non-RS type MDS code (or simply non-RS MDS code). It is well known that any $[n,k]$ MDS code is equivalent to an RS code for $k < 3$ or $n-k< 3$ \cite{bee2017}, and the dual of an RS code is also an RS code. Thus, it is interesting and important to construct non-RS MDS codes in coding theory. However, there are many works focusing on restricting certain parameters for constructing non-RS MDS codes in the literature. For example, Roth and Lempel \cite{roth1989} constructed a class of non-RS MDS codes for even $q$ and $3\leq k \leq q/2-1$ (or $3\leq k \leq q-1$) by using generator polynomials and subsets of the finite field $\mathbb{F}_q$. Recently, Chen \cite{chen2024} constructed some non-RS MDS codes using algebraic curves.
Subsequently, Li et al. \cite{chen2025} obtained some cyclic MDS codes by solving systems of polynomial equations and judged whether these cyclic MDS codes are non-RS by using the Schur square criterion \cite{chen2025}. Very recently, Jin et al. \cite{jin2026} proposed a method for constructing non-RS MDS codes by selecting appropriate polynomials and point sets, and the parameters of these codes are more flexible compared to those of \cite{chen2024,chen2025}. In fact, these construction methods are not only relatively complex in computation but also impose strict restrictions on certain parameters. Motivated by these facts, we investigate this problem further by using
norm reduction in cyclotomic fields and present a construction method of cyclic MDS codes over finite fields via cyclotomic field reduction. We transform the problem of verifying the MDS property over a finite field into a problem of determining non-zero minors in characteristic zero. Specifically, while the method of \cite{chen2025} requires solving polynomial equations and computing Schur square dimensions for restricted defining sets with the form $J=\{0,1,\dots,k\}\setminus\{s\}$, our method only requires computing the prime factors of the algebraic norms of finitely many minors once. The main results of this paper are summarized as follows:
\begin{enumerate}
    \item We establish a lifting-and-reduction framework (see Theorems \ref{finitelift}--\ref{mds2}) that characterizes the MDS property of cyclic codes over finite fields via  a defining set $J$  and an explicitly computable finite set $P_{\text{bad}}$ of rational primes.
    \item When the code length $n \ge 5$ is a prime, for any sufficiently large prime  $p$  such that $n \mid p-1$ (there exist infinitely many such primes), reducing the entries of $G$ modulo a prime ideal $\mathfrak{p}$ of $\mathbb{Z}[\zeta_n]$ lying above $p$ to obtain a matrix $\overline{G}$ over the finite field $\mathbb{F}_p$. Then the code $\overline{\mathcal{C}}$ generated by $\overline{G}$ is a cyclic MDS code of length $n$ and dimension $k$ over $\mathbb{F}_p$ (see Theorem \ref{npri-mds-lift}).
    \item We give an exact, field-independent combinatorial criterion for determining whether some cyclic codes are non-RS (see Theorem \ref{non-rs}).
    \item For the defining set $J = \{0, 1, 2, 4, \dots, 2^{k-2}\}$ with $k \ge 4$ and $2^{k-2} < n/2$, we establish a necessary and sufficient condition for the reduced code $\overline{\C}$ of length $n = 2^s + 1$ ($s \ge 3$) over $\mathbb{F}_{2^{2s}}$ to be an MDS code, which is further shown to be of non-RS type. In particular, we obtain some infinite families of non-RS cyclic MDS codes over $\mathbb{F}_{2^{2s}}$ (see Theorem \ref{binary-rs}).
\end{enumerate}

The rest of this paper is arranged as follows. Section \ref{sec-pre} introduces some notation and results
related to the Schur square criterion and cyclotomic fields which will be needed in subsequent sections. Section \ref{sec-codes} gives a construction method of cyclic MDS codes over finite fields via cyclotomic field reduction and derives many non-RS cyclic MDS codes. Section \ref{sec-summary} concludes this paper and makes concluding remarks.

\section{Preliminaries}\label{sec-pre}

In this section, we briefly recall some results on the Schur square criterion, cyclotomic fields, rings of integers, residue fields and reduction homomorphisms. These results will be used later in this paper.

For a linear code $\C \subseteq \mathbb{F}_q^n$, its Schur square is defined as $\C^2 := \langle \mathbf{c} \ast \mathbf{c}' \mid \mathbf{c}, \mathbf{c}' \in \C \rangle$, where $\mathbf{c} \ast \mathbf{c}' = (c_0 c'_0, c_1 c'_1, \dots, c_{n-1} c'_{n-1})$ denotes the component-wise (Schur) product. The following criterion characterizes when an MDS code is a GRS code (and hence of RS type) via the dimension of its Schur square (see \cite{chen2025}).

\begin{lemma} \cite{chen2025}\label{judg}
Let $\C\subseteq \mathbb{F}_q ^{~n}$ be an MDS code with length $n$ and dimension $k\leq(n-1)/2$. Then the code $\C$ is GRS (i.e., of RS type) if and only if the dimension of the Schur square of the code $\C$ is $2k-1$, i.e., $\dim(\C ^2)=2k-1$.
\end{lemma}

Let $n$ be a positive integer and $\zeta_n = e^{2\pi i/n}$ be a primitive $n$-th root of unity over the complex numbers $\mathbb{C}$. The field $K = \mathbb{Q}(\zeta_n)$ is called the $n$-th cyclotomic field. It is a finite extension of $\mathbb{Q}$ with degree $[K:\mathbb{Q}] = \varphi(n)$, where $\varphi$ is Euler's totient function. Note that the element $\zeta_n$ is an algebraic integer, as it is a root of the monic polynomial $x^n-1=0$.

The set of all algebraic integers in $K$ forms a ring, which is called the ring of integers of $K$ and denoted by $\mathcal{O}_K$. For the cyclotomic field $\mathbb{Q}(\zeta_n)$, the corresponding ring of integers is $\mathcal{O}_K = \mathbb{Z}[\zeta_n]$.

Let $p$ be a rational prime. Then the principal ideal $p\mathcal{O}_K$ is generated by $p$ in $\mathcal{O}_K$. However, the ideal $p\mathcal{O}_K$ is possibly not prime and can be decomposed into a product of prime ideals, i.e.,
\begin{eqnarray}\label{eq-primp}
p\mathcal{O}_K= \mathfrak{p}_1^{e_1}\mathfrak{p}_2^{e_2}\cdots \mathfrak{p}_g ^{e_g},
\end{eqnarray}
where each $\mathfrak{p}_i$ is a prime ideal of $\mathcal{O}_K$ lying over $p$ and each $e_i$ is a positive integer.

For any prime ideal $\mathfrak{p} \subset \mathcal{O}_K$ in (\ref{eq-primp}), $\mathfrak{p} \cap \mathbb{Z} =p \mathbb{Z}$ (i.e., $\mathfrak{p}$ lies over $p$) and the quotient ring $\mathbb{F}_{\mathfrak{p}} = \mathcal{O}_K / \mathfrak{p}$ is a field, called the residue field of $\mathfrak{p}$. This field has characteristic $p$ and is a finite extension of $\mathbb{F}_p$. Let $f = [\mathbb{F}_{\mathfrak{p}} : \mathbb{F}_p]$ denote the degree of this extension. Then we have $\mathbb{F}_{\mathfrak{p}} \cong \mathbb{F}_{p^f}$ and the following property in Lemma \ref{lem-CF}.

\begin{lemma} \label{lem-CF}
Let $p$ be a prime with $p \nmid n$. Let $f = \ord_n(p)$ be the multiplicative order of $p$ modulo $n$. Then, the ideal $p\mathcal{O}_K$ can be decomposed into a product of $\varphi(n)/f$ distinct prime ideals in the cyclotomic field $K = \mathbb{Q}(\zeta_n)$. Furthermore, the residue field degree is $f$ for each such prime ideal $\mathfrak{p}$, which means that $\mathcal{O}_K/\mathfrak{p} \cong \mathbb{F}_{p^f}$.
\end{lemma}

We remark that the above property guarantees the existence of a primitive $n$-th root of unity in $\mathbb{F}_{p^f}$, since $n \mid p^f-1$ and
the multiplicative group $\mathbb{F}_{p^f} ^\times$ is cyclic of order $p^f-1$. This is the theoretical foundation for reducing $\zeta_n$ to an element of a finite field.

For each prime ideal $\mathfrak{p}$ of $\mathcal{O}_K$ lying over $p$  in Lemma \ref{lem-CF}, there exists a natural reduction homomorphism
$$\rho: \mathcal{O}_K \longrightarrow \mathcal{O}_K/\mathfrak{p} \cong \mathbb{F}_{p^f}.$$
This map sends algebraic integers to elements of a finite field. In particular, $\rho(\zeta_n)$ is a primitive $n$-th root of unity in $\mathbb{F}_{p^f}$.

The following theorem of Chebotarev \cite{StevenhagenLenstra1996, zhang2019} will be used later in this paper.

\begin{theorem}[Chebotarev's Theorem]\label{thm:chebotarev}
Let the symbols and notations be as above. Let $n$ be a prime number. Define the $n \times n$ Fourier matrix $V$ by its $(i, j)$-entry
\[
V_{i,j} = \zeta_n ^{\,ij}, \qquad 0 \le i,j \le n-1.
\]
Then all square submatrices of $V$ have non-zero determinant.
\end{theorem}

\section{Non-RS cyclic MDS codes with flexible parameters }\label{sec-codes}

In this section, we give a construction of cyclic MDS codes over finite fields via cyclotomic field reduction. Based on this construction, many non-RS cyclic MDS codes are derived.
Before giving and proving the main results of this paper, we first
prove a few auxiliary results which will be used later.

\subsection{Lifting and reduction principles}

In order to construct MDS codes over finite fields via cyclotomic field reduction in this paper, we need the help of some results that are described and proved in this subsection.

We first show that MDS codes over finite fields with a specific group structure can be lifted to cyclotomic fields in Theorem \ref{finitelift}.
\begin{theorem} \label{finitelift}
Let $\C'$ be an $[N, k]$ MDS code over a finite field $\mathbb{F}_q$. Let $G' = [a'_{i,j}]$ be a generator matrix of $\C'$. Let $n$ be the order of the multiplicative group generated by all non-zero entries of $G'$. Let $\zeta_n' \in \mathbb{F}_q$ be a primitive $n$-th root of unity, such that any non-zero entry $a'_{i,j}$ can be uniquely expressed as $(\zeta_n')^{e_{i,j}}$ for some integer $e_{i,j} \in \{0,1,\dots,n-1\}$.

Define a matrix $G = [a_{i,j}]$ over the $n$-th cyclotomic field $K = \mathbb{Q}(\zeta_n)$ by the following mapping:
$$
a_{i,j} =
\begin{cases}
0, & \text{if } a'_{i,j} = 0 \text{ in } \mathbb{F}_q, \\
\zeta_n^{e_{i,j}}, & \text{if } a'_{i,j} = (\zeta_n ')^{e_{i,j}} \text{ in } \mathbb{F}_q.
\end{cases}
$$
Then, the code $\C$ generated by the matrix $G$ is an $[N, k]$ MDS code over $K$.
\end{theorem}

\begin{proof}
Since $G'$ is the generator matrix of the MDS code $\C'$ over $\mathbb{F}_q$, the determinant $\det(G'_S) \neq 0$ for any $k \times k$ submatrix $G'_S$ with the column index set $S \subseteq \{1,\dots,N\}$, $|S|=k$.

Since $\zeta_n' \in \mathbb{F}_q$ is a primitive $n$-th root of unity, we have $\gcd(n, \mathrm{char}(\mathbb{F}_q)) = 1$, and the minimal polynomial of $\zeta_n$ over $\mathbb{Q}$ (the $n$-th cyclotomic polynomial $\Phi_n(x)$) has $\zeta_n'$ as a root in $\mathbb{F}_q$. Therefore, there exists a well-defined ring homomorphism
\[
\rho: \mathbb{Z}[\zeta_n] \longrightarrow \mathbb{F}_q
\]
uniquely determined by $\rho(1) = 1_{\mathbb{F}_q}$ and $\rho(\zeta_n) = \zeta_n'$. Clearly, $\rho$ maps $0$ to $0$ and $\zeta_n^{e_{i,j}}$ to $(\zeta_n')^{e_{i,j}} = a'_{i,j}$, so that $G' = \rho(G)$ entrywise.

Since the determinant is a polynomial with integer coefficients in the entries of the matrix and $\rho$ is a ring homomorphism, for every $k$-subset $S \subseteq \{1,\dots,N\}$ we have
\[
\rho(\det(G_S)) = \det(\rho(G_S)) = \det(G'_S) \neq 0 \quad \text{in } \mathbb{F}_q.
\]
If $\det(G_S) = 0$ in $K$, then $\rho(\det(G_S)) = \rho(0) = 0$ in $\mathbb{F}_q$, which contradicts $\det(G'_S) \neq 0$. Hence $\det(G_S) \neq 0$ in $K$.
Thus, for every $k$-subset $S$ of the column index set $\{1,\dots,N\}$ of $G$, the corresponding $k\times k$ submatrix $G_S$ has nonzero determinant over $K$. Hence $G$ has full rank $k$
and any $k$ columns of $G$ are linearly independent. Therefore, the code $\C$ generated by the matrix $G$ is an $[N,k]$ MDS code over the cyclotomic field $K$.
\end{proof}

Recall that $K = \mathbb{Q}(\zeta_n)$ is the $n$-th cyclotomic field and the corresponding ring of integers is $\mathcal{O}_K = \mathbb{Z}[\zeta_n]$ for the cyclotomic field $K$.
We have the results via cyclotomic field reduction in the following theorem.

\begin{theorem}\label{reduction}
Let the symbols and notations be as above. Let $\mathcal{C}$ with the generator matrix $G \in M_{k \times N}(\mathcal{O}_K)$ be an $[N, k]$ MDS code over $K$. Let
\[
D := \prod_{\substack{S \subseteq \{1,\dots,N\} \\ |S| = k}} \det(G_S) \in \mathcal{O}_K \setminus \{0\},
\]
and let $P_G := \{ p \text{ rational prime} \mid p \text{ divides } |N_{K/\mathbb{Q}}(D)| \}$ be a finite set of rational primes. Then for any rational prime $p$ satisfying $p \nmid n$ and $p \notin P_G$, and for any prime ideal $\mathfrak{p}$ of $\mathcal{O}_K$ lying over $p$, the reduced matrix $\overline{G} = G \pmod{\mathfrak{p}}$ generates an $[N, k]$ MDS code $\overline{\C}$ over the finite field $\mathbb{F}_q \cong \mathcal{O}_K / \mathfrak{p}$, where $q=p^f$ and $f = \ord_n(p)$ is the multiplicative order of $p$ modulo $n$.
\end{theorem}

\begin{proof}
Since $G$ is the generator matrix of the $[N,k]$ MDS code $\C$, the determinant $\det(G_S) \neq 0$ for any $k \times k$ submatrix $G_S$ with the column index $k$-subset $S \subseteq \{1,\dots,N\}$.
As all entries of $G$ lie in $\mathcal{O}_K$,
we have
\[
\Delta_S := \det(G_S) \in \mathcal{O}_K \setminus \{0\}.
\]
Since $D = \prod_{|S|=k} \Delta_S \neq 0$, its algebraic norm $N_{K/\mathbb{Q}}(D)$ is a non-zero rational integer, so $P_G$ is a finite set of rational primes. For any rational prime $p \notin P_G$ with $p \nmid n$ and any prime ideal $\mathfrak{p}$ of $\mathcal{O}_K$ lying over $p$, we have $\Delta_S \not\equiv 0 \pmod{\mathfrak{p}}$ for each $k$-subset $S \subseteq \{1,\dots,N\}$; otherwise $\mathfrak{p} \mid \Delta_S$ would imply $p \mid N_{K/\mathbb{Q}}(\Delta_S) \mid N_{K/\mathbb{Q}}(D)$, contradicting $p \notin P_G$.

By Lemma \ref{lem-CF}, since $p \nmid n$, the residue field $\mathcal{O}_K/\mathfrak{p}$ is isomorphic to $\mathbb{F}_q$ with $q = p^f$ and $f = \ord_n(p)$. Consider the reduction homomorphism $$\rho: \mathcal{O}_K \to \mathcal{O}_K/\mathfrak{p} \cong \mathbb{F}_q.$$ Let $\overline{G} = \rho(G)$ be the matrix obtained by reducing each entry modulo $\mathfrak{p}$.
Then
\[
\det(\overline{G}_S) = \rho(\det(G_S)) = \rho(\Delta_S) \neq 0 \quad \text{in } \mathbb{F}_q
\]
for any $k$-subset $S \subseteq \{1,\dots,N\}$, since $\mathfrak{p} \nmid \Delta_S$. Thus, every $k$ columns of $\overline{G}$ are linearly independent over $\mathbb{F}_q$. This
means that $\overline{G}$ generates an $[N,k]$ MDS code $\overline{\C}$ over $\mathbb{F}_q$.
\end{proof}

We remark that the results of Theorem \ref{reduction} show that if there exists an $[N, k]$ MDS code over $K$, then there must exist an $[N, k]$ MDS code over the finite field $\mathbb{F}_q$ for infinitely many prime powers $q$.
Next we will give some cyclic MDS codes over finite fields via cyclotomic field reduction.

\subsection{Constructions of cyclic MDS codes}
Let $n\geq 4$ be a positive integer. Define a set $J=\{j_1,j_2,\dots,j_k\} \subseteq \{0,1,2,\dots,n-1\}$ and construct a $k \times n$ matrix $G$
with $(i, l)$-entry
\[
G_{i,l} = \zeta_n^{j_i(l-1)}, \qquad 1\le i\le k,\; 1\le l\le n,
\]
i.e.,
\begin{eqnarray}\label{eq-G}
G=
\begin{bmatrix}
1 & \zeta_n^{j_1 \cdot 1} & \ldots & \zeta_n^{j_1 \cdot (n-1)} \\
1 & \zeta_n^{j_2 \cdot 1} & \ldots & \zeta_n^{j_2 \cdot (n-1)} \\
\vdots & \vdots & \vdots & \vdots  \\
1 & \zeta_n^{j_k \cdot 1} & \ldots & \zeta_n^{j_k \cdot (n-1)} \\
\end{bmatrix}.
\end{eqnarray}
Without loss of generality, we assume that $j_1<j_2<\dots<j_k$. It is clear that the matrix $G$ can be viewed as the generator matrix (or parity check matrix) of a code defined over the cyclotomic field $K = \mathbb{Q}(\zeta_n)$, and the set $J$ is called the defining set of this code.

Note that a linear code $\C$ of length $n$ and dimension $k$ is MDS if its minimum Hamming distance $d = n - k + 1$. This means that the code $\C$ is MDS if and only if every $k \times k$ submatrix of its generator matrix is non-singular.

For a fixed pair $(n, J)$, suppose that for every $k$-subset $S \subseteq \{1,\dots,n\}$, the $k \times k$ submatrix $G_S$ has a non-zero determinant $\Delta_S := \det(G_S) \in \mathcal{O}_K \setminus \{0\}$. We define the finite set $P_{\text{bad}}$ of ``bad'' rational primes as the union of all rational prime factors dividing the absolute algebraic norms of these determinants:
\begin{eqnarray}\label{eq-pbad}
P_{\text{bad}} = \bigcup_{\substack{S \subseteq \{1,\dots,n\} \\ |S|=k}} \left\{ p \in \mathbb{Z} \mid p \text{ is a prime and } p \mid |N_{K/\mathbb{Q}}(\Delta_S)| \right\},
\end{eqnarray}
where $N_{K/\mathbb{Q}}(\cdot)$ is the norm mapping from $K$ to $\mathbb{Q}$. Combining the reduction characterization of minors with Theorem \ref{reduction}, we obtain the following theorem.

\begin{theorem}\label{mds2}
Let the symbols and notations be as above. Let $n\geq 4$ be a positive integer and $J = \{j_1, \dots, j_{k}\} \subseteq \{0,1, \dots, n-1\}$ with $k\geq 3$. Assume that for every $k$-subset $S \subseteq \{1,\dots,n\}$, the $k\times k$ submatrix $G_S$ of the matrix $G$ in (\ref{eq-G}) has non-zero determinant $\Delta_S = \det(G_S)$ in the cyclotomic field $K$. Let $p$ be a rational prime with $p \nmid n$, and suppose that the principal ideal $p\mathcal{O}_K$ decomposes as
\begin{eqnarray}\label{pl}
p\mathcal{O}_K= \prod_{1\leq l \leq g} \mathfrak{p}_l
\end{eqnarray}
in the cyclotomic field $K$, where each prime ideal $\mathfrak{p}_l$ corresponds to a residue field $\mathcal{O}_K / \mathfrak{p}_l \cong \mathbb{F}_{p^f}$ and $f = \ord_n(p)$ is the multiplicative order of $p$ modulo $n$. For any prime ideal $\mathfrak{p}$ in (\ref{pl}), we reduce the entries of matrix $G$ modulo $\mathfrak{p}$ and obtain a matrix $$\overline{G} = \rho(G) \in \mathbb{F}_q^{k \times n}$$ over the finite field $\mathbb{F}_q$, where $q=p^f$. Then we have the following results.
\begin{enumerate}
\item The code $\overline{\C}$ generated by $\overline{G}$ is a cyclic MDS code over $\mathbb{F}_q$ if and only if $\Delta_S \not\equiv 0 \pmod{\mathfrak{p}}$ for all $k$-subsets $S \subseteq \{1,\dots,n\}$.
    \item
    In particular, if $p \notin P_{\text{bad}}$ (with $p \nmid n$), then for every prime ideal $\mathfrak{p}$ lying above $p$, the code $\overline{\C}$ generated by $\overline{G}$ is a cyclic MDS code over $\mathbb{F}_q$ with length $n$, dimension $k$ and minimum distance $d = n - k + 1$. Conversely, if $p \in P_{\text{bad}}$ (with $p \nmid n$), then for every prime ideal $\mathfrak{p}$ lying above $p$, $\overline{\C}$ is not an MDS code over $\mathbb{F}_q$.
\end{enumerate}
\end{theorem}

\begin{proof}

1) By definition, for any prime $p \nmid n$ and $\mathfrak{p} \mid p\mathcal{O}_K$, $\overline{\zeta}_n := \rho(\zeta_n)$ is a primitive $n$-th root of unity in $\mathbb{F}_q$. Each row of $\overline{G}$ is of the form $(1, \overline{\zeta}_n^{j_i}, \overline{\zeta}_n^{2j_i}, \dots, \overline{\zeta}_n^{(n-1)j_i})$. A cyclic right shift of the $i$-th row multiplies the row by $\overline{\zeta}_n^{-j_i} \in \mathbb{F}_q^\times$ because $\overline{\zeta}_n^n = 1$, and therefore leaves the row space $\overline{\C}$ invariant. Hence $\overline{\C}$ is a cyclic code of length $n$. Moreover, $\overline{\C}$ is an $[n,k]$ MDS code if and only if every $k \times k$ minor $\det(\overline{G}_S) = \rho(\Delta_S)$ is non-zero in $\mathbb{F}_q$, which is equivalent to $\Delta_S \not\equiv 0 \pmod{\mathfrak{p}}$ for all $k$-subsets $S \subseteq \{1,\dots,n\}$.

2) Note that the non-zero determinant $\Delta_S \neq 0$ over characteristic $0$ implies that $|N_{K/\mathbb{Q}}(\Delta_S)| \geq 1$. When reducing modulo $\mathfrak{p}$ for a given $p \notin P_{\text{bad}}$ with $p \nmid n$, if there exists a $k\times k$ submatrix of $\overline{G}$ such that its determinant $\overline{\Delta}_S = 0$ in $\mathbb{F}_q$, then the corresponding $\Delta_S \equiv 0 \pmod{\mathfrak{p}}$. This means that $p \mid N_{K/\mathbb{Q}}(\Delta_S)$ and thus $p \in P_{\text{bad}}$, which contradicts $p \notin P_{\text{bad}}$. Thus, every $k \times k$ submatrix of the generator matrix $\overline{G}$ of $\overline{\C}$ is non-singular over $\mathbb{F}_q$. Hence $\overline{\C}$ has dimension $k$ and minimum distance $d = n - k + 1$, i.e., the Singleton bound is attained.
Conversely, suppose $p \in P_{\text{bad}}$ with $p \nmid n$, so that $p \mid N_{K/\mathbb{Q}}(\Delta_S)$ for some $k$-subset $S$ and hence $\Delta_S \equiv 0 \pmod{\mathfrak{p}_l}$ for some prime ideal $\mathfrak{p}_l$ lying above $p$. Since $\mathrm{Gal}(K/\mathbb{Q})$ acts transitively on the prime ideals lying above $p$, and each automorphism $\sigma_a \in \mathrm{Gal}(K/\mathbb{Q})$ given by $\sigma_a(\zeta_n) = \zeta_n^a$ with $\gcd(a,n) = 1$ permutes the columns of $G$, all prime ideals above $p$ yield the same pattern of vanishing minors. Therefore, for every prime ideal $\mathfrak{p}$ lying above $p$, there exists a $k$-subset $S'$ such that $\Delta_{S'} \equiv 0 \pmod{\mathfrak{p}}$, and thus $\overline{\C}$ is not an MDS code over $\mathbb{F}_q$.
This completes the proof.
\end{proof}

We remark that for a given pair $(n, J)$, Theorem \ref{mds2} completely determines whether $\overline{\C}$ is a cyclic MDS code for every prime $p \nmid n$ and every prime ideal $\mathfrak{p}$ lying above $p$. Specifically, $\overline{\C}$ is an MDS code if and only if $p \notin P_{\text{bad}}$ and $p \nmid n$, regardless of the choice of the prime ideal $\mathfrak{p}$ above $p$. The key is to find the set $P_{\text{bad}}$ in (\ref{eq-pbad}). Computing $P_{\text{bad}}$ is straightforward in principle, as it requires traversing all $\binom{n}{k}$ minors of size $k \times k$ once and factoring their algebraic norms. Based on Theorem \ref{mds2}, we give a method which consists of three steps as follows: \begin{description}
\item[\textbf{Step 1:}]~Given $n\geq 4$ and $J = \{j_1, \dots, j_{k}\} \subseteq \{0,1, \dots, n-1\}$ with $k\geq 3$.
\item[\textbf{Step 2:}]~Construct the $k \times n$ matrix $G$ over $K$ by (\ref{eq-G}).
\item[\textbf{Step 3:}]~Traverse all $k\times k$ submatrices of the matrix $G$, performing the following operations:
\begin{itemize}
\item For each $k\times k$ submatrix $G_S$ of $G$, compute its determinant $\Delta_S$. If there exists a $k\times k$ submatrix whose determinant is zero, then no cyclic MDS code over any finite field can be obtained from this pair $(n,J)$, and we terminate. Otherwise, the determinants of all $k\times k$ submatrices are non-zero, and we proceed to the next step.
    \item For each non-zero determinant $\Delta_S$, put all rational prime factors of $|N_{K/\mathbb{Q}}(\Delta_S)|$ into the set $P_{\text{bad}}$.
    Then $\overline{\C}$ is a cyclic MDS code over $\mathbb{F}_q$ for every prime $p \notin P_{\text{bad}}$ with $p \nmid n$ and every prime ideal $\mathfrak{p}$ lying above $p$.
\end{itemize}
\end{description}

Let the symbols and notations be as above. Using the above method to find $P_{\text{bad}}$, we give some examples in Table \ref{tab1} which are confirmed by Magma programs, where ``---'' indicates that there exists zero minors and the algorithm terminates.
\begin{table}[h]

\begin{center}

\caption{Some examples for finding the set $P_{\text{bad}}$. \label{tab1}}
\begin{tabular}{ccccc} \hline
$n$ & $J$  &   $k$      & \makecell{All $k\times k$ minors \\are non-zero}            &$P_{\text{bad}}$   \\ \hline
$7$  & $\{0,1,3\}$  & $3$        & Yes          &$\{2, 7   \}$   \\  [2mm]

$7$  & $\{0,1,4\}$  & $3$        & Yes          &$\{7   \}$   \\  [2mm]

$7$  & $\{0,2,3\}$  & $3$        & Yes          &$\{2,7  \}$   \\  [2mm]

$8$  & $\{0,2,3\}$  & $3$        & Yes          &$\{2,3   \}$   \\  [2mm]

$9$  & $\{0,2,3\}$  & $3$        & No          & ---   \\  [2mm]

$9$  & $\{0,1,5\}$  & $3$        & Yes          &$\{3   \}$   \\  [2mm]

$9$  & $\{2,3,4\}$  & $3$        & Yes          &$\{3   \}$   \\  [2mm]

$12$  & $\{2,3,4\}$  & $3$        & Yes          &$\{2, 3   \}$   \\  [2mm]

$12$  & $\{0,1,4\}$  & $3$        & No          & ---  \\  [2mm]

$12$  & $\{0,2,7\}$  & $3$        & Yes          &$\{2, 3   \}$   \\  [2mm]

$12$  & $\{0,1,2,5\}$  & $4$        & No         & ---  \\  [2mm]

$13$  & $\{0,1,2,4\}$  & $4$     & Yes          &$\{3, 13, 53, 79, 157   \}$   \\  [2mm]

$13$  & $\{0,1,3,6\}$  & $4$     & Yes          &$\{3, 5, 13, 53, 521, 1327  \}$   \\  [2mm]

$15$  & $\{0,1,2,4\}$  & $4$     & Yes          &$\{ 2, 3, 5, 11, 31, 61, 211  \}$   \\  [2mm]

$18$  & $\{0,1,2,4\}$  & $4$     & No          & ---  \\  [2mm]

$18$  & $\{0,1,5,8\}$  & $4$     & Yes          &\makecell{ $\{ 2, 3, 17, 19, 37, 53, 73, 127, $\\$ 163, 181, 397, 631, 757, 2089, 17137   \}$ }  \\  [2mm]

$23$  & $\{0,1,2,4\}$  & $4$     & Yes          &\makecell{$\{ 23, 47, 139, 277, 461, 691, 1289, 2393, $\\$  3037, 5107, 6763, 11593, 14537, 102397   \}$ }  \\  [2mm]
\hline
\end{tabular}
\end{center}
\end{table}

Next, we establish an exact, field-independent criterion for determining whether the cyclic MDS code $\overline{\C}$ obtained in Theorem \ref{mds2} is of RS type or non-RS type. For a subset $J \subseteq \{0,1,\dots,n-1\}$, we denote its sumset modulo $n$ by
\[
J + J \bmod n := \{ (a + b) \bmod n \mid a, b \in J \} \subseteq \{0,1,\dots,n-1\}.
\]

\begin{lemma}[Schur square of cyclic evaluation codes]\label{lem:schur-cyclic}
Let $F$ be any field containing a primitive $n$-th root of unity $\zeta$ (in particular, $\mathrm{char}(F) \nmid n$), and let $J \subseteq \{0,1,\dots,n-1\}$ with $|J| = k \ge 3$. Let $\C_F = \langle \mathrm{ev}_j \mid j \in J \rangle_F \subseteq F^n$, where $\mathrm{ev}_j = (\zeta^{j \cdot 0}, \zeta^{j \cdot 1}, \dots, \zeta^{j \cdot (n-1)})$. Then
\[
\dim_F(\C_F^2) = |J + J \bmod n|.
\]
Consequently, when $\C_F$ is an $[n,k]$ MDS code with $k \le (n-1)/2$, $\C_F$ is of RS type if and only if $|J + J \bmod n| = 2k - 1$; in particular, $\C_F$ is of non-RS type if and only if $|J + J \bmod n| \ge 2k$. Thus, the RS or non-RS type of $\C_F$ depends only on $(n, J)$ and is independent of the underlying field $F$ (so $\C$ over $K$ and $\overline{\C}$ over $\mathbb{F}_q$ are simultaneously of RS type or simultaneously of non-RS type).
\end{lemma}

\begin{proof}
The Schur square $\C_F^2$ is spanned over $F$ by the component-wise products $\mathrm{ev}_a \ast \mathrm{ev}_b$ for $a, b \in J$. Since $\zeta^n = 1$, we have $\mathrm{ev}_a \ast \mathrm{ev}_b = \mathrm{ev}_{(a+b)\bmod n}$, and hence
\[
\C_F^2 = \langle \mathrm{ev}_c \mid c \in J + J \bmod n \rangle_F.
\]
To see that the vectors $\{\mathrm{ev}_c \mid c \in J + J \bmod n\}$ are linearly independent over $F$, suppose that $\sum_{c \in S} \lambda_c \mathrm{ev}_c = \mathbf{0}$ for some coefficients $\lambda_c \in F$, where $S = J + J \bmod n \subseteq \{0,1,\dots,n-1\}$. Then the polynomial $P(X) = \sum_{c \in S} \lambda_c X^c \in F[X]$ has degree less than $n$ and vanishes at all $n$ distinct powers $1, \zeta, \zeta^2, \dots, \zeta^{n-1}$. Hence $P(X)$ is the zero polynomial, which implies $\lambda_c = 0$ for all $c \in S$. Therefore, $\dim_F(\C_F^2) = |S| = |J + J \bmod n|$. Combining this with Lemma \ref{judg} yields the desired equivalence for $k \le (n-1)/2$ (note that since $\C_F$ is MDS of dimension $k \le (n-1)/2$, we always have $\dim_F(\C_F^2) \ge 2k - 1$).
\end{proof}

\begin{remark}\label{rem:rs-nonrs}
By Lemma \ref{lem:schur-cyclic}, whether the code $\overline{\C}$ obtained by Theorem \ref{mds2} (for $k \le (n-1)/2$) is an RS code or a non-RS code is completely determined by $(n, J)$ via the cardinality $|J + J \bmod n|$. Specifically, $\overline{\C}$ is of RS type if and only if $|J + J \bmod n| = 2k - 1$. In particular, when $J$ is an arithmetic progression, $\overline{\C}$ is of RS type; however, if $J$ is not an arithmetic progression in $\mathbb{Z}$, $\overline{\C}$ may still be of RS type if modular wrap-around modulo $n$ collapses $J + J \bmod n$ to $2k - 1$ elements (for example, when $n = 7$ and $J = \{0, 1, 4\}$, $J + J = \{0, 1, 2, 4, 5, 8\}$ has $6$ integers, but $J + J \bmod 7 = \{0, 1, 2, 4, 5\}$ has $2k - 1 = 5$ elements, yielding an RS code). We give some examples confirmed by Magma programs in Table \ref{tab-rs-nonrs}.
\end{remark}

\begin{table}[H]

\begin{center}
\caption{Some examples on Theorem \ref{mds2} and Lemma \ref{lem:schur-cyclic}. \label{tab-rs-nonrs}}
\begin{tabular}{cccccccc} \hline
$n$ & $J$  &   $k$      & $|J+J\bmod n|$ & \makecell{All $k\times k$ minors \\are non-zero}      &\makecell{Is $\overline{\C}$ a non-RS \\cyclic MDS code?} \\ \hline
$7$  & $\{0,1,3\}$  & $3$   & $6$     & Yes          & Yes   \\  [2mm]

$7$  & $\{0,2,3\}$  & $3$   & $6$     & Yes          & Yes   \\  [2mm]

$7$  & $\{0,1,4\}$  & $3$   & $5$     & Yes          & No    \\  [2mm]

$7$  & $\{0,2,5\}$  & $3$   & $5$     & Yes          & No    \\  [2mm]

$7$  & $\{0,1,5\}$  & $3$   & $6$     & Yes          &Yes   \\  [2mm]

$7$  & $\{1,2,5\}$  & $3$   & $5$     & Yes          &No   \\  [2mm]

$8$  & $\{1,2,4\}$  & $3$   & $6$     & Yes          &Yes   \\  [2mm]

$9$  & $\{0,2,7\}$  & $3$   & $5$     & Yes          &No   \\  [2mm]

$10$  & $\{0,1,4\}$  & $3$  & $6$      & Yes          &Yes  \\  [2mm]

$13$  & $\{0,1,3,4\}$  & $4$ & $9$       & Yes          &Yes   \\  [2mm]

\hline
\end{tabular}
\end{center}
\end{table}

Note that for any prime $n$ and any set $J = \{j_1, j_2, \dots, j_k\} \subseteq \{0,1,\dots,n-1\}$, we can construct cyclic MDS codes which are described in Theorem \ref{npri-mds-lift}.
\begin{theorem}
\label{npri-mds-lift}
Let the symbols and notations be as above. Let $n \ge 5$ be a prime number and let $J = \{j_1, j_2, \dots, j_k\} \subseteq \{0,1,\dots,n-1\}$ with $k \ge 3$.
Construct the $k \times n$ matrix $G$ over the cyclotomic field $K = \mathbb{Q}(\zeta_n)$ by
\[
G_{i,l} = \zeta_n^{j_i(l-1)}, \qquad 1\le i\le k,\; 1\le l\le n.
\]
We have the following results.

\begin{enumerate}
    \item The code $\mathcal{C}$ generated by $G$ is a cyclic MDS code of length $n$ and dimension $k$ over $K$.
    \item For any sufficiently large prime  $p$  such that $n \mid p-1$ (there exist infinitely many such primes), reduce the entries of $G$ modulo a prime ideal $\mathfrak{p}$ of $\mathbb{Z}[\zeta_n]$ lying above $p$ to obtain a matrix $\overline{G}$ over the finite field $\mathbb{F}_p$. Then the code $\overline{\mathcal{C}}$ generated by $\overline{G}$ is a cyclic MDS code of length $n$ and dimension $k$ over $\mathbb{F}_p$.
\end{enumerate}
\end{theorem}

\begin{proof}
By definitions, it is clear that the code $\C$ is a cyclic code with length $n$ over $K$. Since $n$ is prime, from Theorem \ref{thm:chebotarev} and the property of MDS codes we have that $\C$ is an MDS code. Thus, the result 1) holds. Since $n$ is prime, the desired conclusions in 2) follow immediately from Theorem \ref{mds2} and the result 1). This completes the proof.
\end{proof}

\begin{example}
Some examples on cyclic MDS codes obtained by Theorem \ref{npri-mds-lift} are given in Table \ref{n-pri}, which are confirmed by Magma programs.
\end{example}

\begin{table}[H]
\begin{center}
\caption{Some examples on Theorem \ref{npri-mds-lift} when $n | p-1.$  \label{n-pri}}
\begin{tabular}{cccccccc} \hline
$n$ & $J$  &                                  $k$      &   $p\leq 100$              & \makecell{ $\C$ is cyclic MDS code \\ over $K$ }           &\makecell{ $\overline{\C}$ is cyclic MDS code \\over $\mathbb{F}_p$} \\ \hline

$5$  & \makecell{$\{0,1,3\}$ \\ $\{0,2,3\}$ \\ $\{0,1,4\}$ \\ $\{0,2,4\}$  }                     & $3$       &  $11, 31, 41, 61, 71$     & Yes          & Yes   \\  [2mm]

$5$  & \makecell{$\{0,1,3,4\}$ \\ $\{0,2,3,4\}$}    & $4$       &  $11, 31, 41, 61, 71$     & Yes          & Yes   \\  [2mm]

$7$  & \makecell{$\{0,1,3\}$ \\ $\{0,2,3\}$ \\ $\{0,1,4\}$ \\ $\{0,2,5\}$  }                     & $3$       &  $ 29, 43, 71$     & Yes          & Yes   \\  [2mm]

$7$  & \makecell{$\{0,1,3,4\}$ \\ $\{0,1,3,5\}$}    & $4$       &  $29, 43, 71$     & Yes          & Yes   \\  [2mm]

\hline
\end{tabular}
\end{center}
\end{table}

\subsection{Many non-RS type cyclic MDS codes}

Note that any $[n,k]$ MDS code is equivalent to an RS code for $k < 3$ or $n-k< 3$. Meanwhile, the dual of an RS code is also an RS code. Thus, we only consider non-RS cyclic MDS codes with $3\leq k\leq n/2$. By appropriately selecting the defining set $J$, cyclic MDS codes constructed by Theorem \ref{mds2} are non-RS codes \cite{chen2025,jin2026}. In particular, by using Lemma \ref{lem:schur-cyclic}, we obtain a complete characterization for $k \le (n-1)/2$, as well as a clean sufficient condition when the elements of $J$ do not form an arithmetic progression and $j_k \le (n-1)/2$, which are described in Theorem \ref{non-rs}.

\begin{theorem}\label{non-rs}
Let the symbols and notations be as above. Let $n\geq 7$ be a positive integer and $J = \{j_1, \dots, j_{k}\} \subseteq \{0,1, \dots, n-1\}$ with $k\geq 3$ and $j_1 < j_2 < \dots < j_k$. Assume that all $k\times k$ submatrices of the matrix $G$ have non-zero determinant in the cyclotomic field $K$. Let $\overline{\C}$ be the cyclic MDS code over $\mathbb{F}_q$ constructed by Theorem \ref{mds2} for a prime $p \notin P_{\text{bad}}$ with $p \nmid n$.
\begin{enumerate}
    \item If $k \le (n-1)/2$, then $\overline{\C}$ is a non-RS code if and only if $|J + J \bmod n| \ge 2k$.
    \item In particular, if $j_k \leq (n-1)/2$ and the elements of $J$ do not form an arithmetic progression in $\mathbb{Z}$, then $k \le (n-1)/2$, $|J + J \bmod n| = |J + J| \ge 2k$, and $\overline{\C}$ is a non-RS code.
\end{enumerate}
\end{theorem}

\begin{proof}
1) The desired conclusions follow directly from Theorem \ref{mds2} and Lemma \ref{lem:schur-cyclic}.

2) Suppose that $j_k \le (n-1)/2$ and $J$ does not form an arithmetic progression. Since $0 \le j_1 < j_2 < \dots < j_k \le \lfloor (n-1)/2 \rfloor$, we have $k \le \lfloor (n-1)/2 \rfloor + 1$. If $k = (n+1)/2$ (when $n$ is odd), then $J=\{0, 1, \dots, (n-1)/2\}$. This means that $J$ forms an arithmetic progression, which contradicts our assumption. Hence, $k \le (n-1)/2$.

Moreover,  since $0 \le a, b \le j_k \le (n-1)/2$ for any  $a, b \in J$,   we have $0 \le a + b \le n - 1$, and hence $|J + J \bmod n| = |J + J|$. By a classical theorem in additive combinatorics (see \cite[Theorem 1.4]{Freiman1973} or \cite[Proposition 5.8]{TaoVu2006}), for any finite subset $J \subset \mathbb{Z}$ of cardinality $k$, one has $|J + J| \ge 2k - 1$, with equality if and only if $J$ is an arithmetic progression. Since $J$ is not an arithmetic progression, we obtain $|J + J \bmod n| = |J + J| \ge 2k$. By the result 1), $\overline{\C}$ is a non-RS code.
\end{proof}

\begin{example}
Some examples on non-RS type cyclic MDS codes obtained by Theorem \ref{non-rs} are given in Table \ref{tab2}, which are confirmed by Magma programs. In particular, the existence of non-RS type cyclic MDS
codes over finite fields with all possible characteristics $p$ can be determined for a given $(n,J)$.
\end{example}

\begin{table}[h]
\begin{center}
\caption{Some examples on non-RS type cyclic MDS codes over finite fields in Theorem \ref{non-rs}. \label{tab2}}
\begin{tabular}{ccccccc} \hline
$n$ & $J$  &   $k$      & \makecell{All $k\times k$ minors \\are non-zero}            &\makecell{$\overline{\C}$ is a  non-RS cyclic MDS code \\ over $\mathbb{F}_q$ with all possible characteristics $p$ } \\ \hline
$7$  & $\{0,1,3\}$  & $3$        & Yes          & $p\notin \{2, 7   \}$   \\  [2mm]

$7$  & $\{0,2,3\}$  & $3$        & Yes          & $p\notin \{2, 7   \}$   \\  [2mm]

$8$  & $\{0,2,3\}$  & $3$        &  Yes         &   $p\notin \{2, 3  \}$ \\  [2mm]

$8$  & $\{0,1,3\}$  & $3$        &  Yes         &   $p\notin \{2, 3  \}$ \\  [2mm]

$9$  & $\{0,2,4\}$  & $3$        & Yes         & RS type  \\  [2mm]

$9$  & $\{0,1,4\}$  & $3$        & No         & Not MDS  \\  [2mm]

$9$  & $\{0,3,4\}$  & $3$        & No         & Not MDS  \\  [2mm]

$9$  & $\{0,1,3\}$  & $3$        & No         & Not MDS  \\  [2mm]

$9$  & $\{0,2,3\}$  & $3$       & No         & Not MDS  \\  [2mm]

$9$  & $\{1,3,4\}$  & $3$         & No         & Not MDS  \\  [2mm]

$9$  & $\{1,2,4\}$  & $3$         & No         & Not MDS  \\  [2mm]

$10$  & $\{1,2,4\}$  & $3$        & Yes         & $p\notin \{2, 5,11  \}$    \\  [2mm]

$10$  & $\{1,3,4\}$  & $3$        & Yes         & $p\notin \{2, 5,11  \}$    \\  [2mm]

$10$  & $\{0,2,3\}$  & $3$        & Yes         & $p\notin \{2, 5,11  \}$    \\  [2mm]

$13$  & $\{0,1,2,5\}$  & $4$        & Yes         &  $p\notin \{3, 13, 157, 521, 599  \}$   \\  [2mm]

$23$  & $\{0,1,2,5\}$  & $4$        & Yes         &~  \makecell{$p\notin \{23, 47, 137, 139, 277, 599, 691, 967, $\\$ 1151, 1933, 15319, 19919, 24841, 53407,
                                                                                        $\\$ 64217, 152767, 677167, 1946767, 1989961 \}$ }   \\  [2mm]

\hline
\end{tabular}
\end{center}
\end{table}

We remark that the defining set $J$ was only considered of the form $J=\{0,1,2,\dots,k\}\setminus \{s\}$ in \cite{chen2025}, where $s\in \{1,2,\dots,k-1\}$. Selecting the same defining sets as those in the $6$ examples of \cite{chen2025}, we easily obtain the corresponding results of \cite{chen2025} by using Theorem \ref{non-rs} of this paper. These corresponding results are described in Table \ref{tab3} and confirmed by Magma programs.

\begin{table}[H]

\begin{center}
\caption{The same $(n,J)$ as those in the $6$ examples of \cite{chen2025}. \label{tab3}}
\begin{tabular}{ccccccc} \hline
$n$ & $J$  &   $k$      & \makecell{All $k\times k$ minors \\are non-zero}            &\makecell{$\overline{\C}$ is a non-RS cyclic MDS code \\ over $\mathbb{F}_q$ with all possible characteristics $p$ }\\ \hline
$7$  & $\{0,1,3\}$ ~or~ $\{0,2,3\}$   & $3$        & Yes          & $p\notin \{2, 7   \}$   \\  [2mm]

$8$  & $\{0,1,3\}$ ~or ~$\{0,2,3\}$  & $3$        & Yes          &$ p \notin  \{2, 3   \}$   \\  [2mm]

$10$  & $\{0,1,3\}$ ~or ~$\{0,2,3\}$  & $3$        & Yes          &$p \notin  \{2, 5, 11   \}$   \\  [2mm]

$20$  & $\{0,1,3\}$ ~or ~$\{0,2,3\}$   & $3$        & Yes          &$ p \notin  \{ 2, 5, 11, 41, 61   \}$   \\  [2mm]

$9$  & $\{0,1,3,4\}$  & $4$        & No          & Not~MDS   \\  [2mm]

$9$  & $\{0,2,3,4\}$  & $4$        & Yes          &$ p \notin  \{ 3,19   \}$   \\  [2mm]

\hline
\end{tabular}
\end{center}
\end{table}

Note that when $n\geq 7$ with $3\nmid n$ and $J=\{0,1,3\}$ (or $J=\{0,2,3\}$), from Theorem \ref{non-rs} we can obtain the results of Theorem 1 of \cite{chen2025}. These results are described in the following corollary.

\begin{corollary}\label{cor:non-rs-mds-3}
Let the symbols and notations be as above. Let $n\geq 7$ be a positive integer with \(3 \nmid n\). Let \(J = \{0,1,3\}\) or \(J = \{0,2,3\}\).
Denote the \(3 \times n\) matrix \(G\) over the cyclotomic field \(\mathbb{Q}(\zeta_n)\) by
\[
G_{i,l} = \zeta_n^{j_i (l-1)},\qquad 1\le i\le 3,\; 1\le l\le n.
\]
Let
\[
P_{\text{bad}} = \bigcup_{|S|=3} \bigl\{ p \text{ prime} \mid p \mid |N_{K/\mathbb{Q}}(\Delta_S)| \bigr\},
\]
where \(\Delta_S\) runs over all determinants of \(3\times 3\) submatrices of \(G\). Let $p$ be a prime with $p \nmid n$ and $p \notin P_{\text{bad}}$, and let \(\mathfrak{p}\) be a prime ideal of \(\mathbb{Z}[\zeta_n]\) lying above \(p\).
Denote by \(\rho: \mathbb{Z}[\zeta_n] \to \mathbb{Z}[\zeta_n]/\mathfrak{p} \cong \mathbb{F}_{p^f}\) the reduction homomorphism and set \(\overline{G} = \rho(G)\). Then the linear code $\overline{\C}$ generated by $\overline{G}$ is a non-RS cyclic MDS code of length \(n\) and dimension \(3\) over \(\mathbb{F}_{p^f}\), where \(f = \ord_n(p)\) is the multiplicative order of \(p\) modulo \(n\).
\end{corollary}

\begin{proof}
For any column index set $S = \{l_1,l_2,l_3\} \subseteq \{1,\dots,n\}$, we set \(x_t = \zeta_n^{l_t-1}\) for $1\leq t \leq 3$. Then these three \(x_t\) are distinct \(n\)-th roots of unity.
The determinant of the corresponding $3\times 3$ submatrix is
\begin{eqnarray}
\Delta_S =
\left\{
\begin{array}{ll}
(x_2-x_1)(x_3-x_1)(x_3-x_2)(x_1+x_2+x_3),            & \text{if } J=\{0,1,3\},\\[2mm]
x_1x_2x_3\,(x_2-x_1)(x_3-x_1)(x_3-x_2)(x_1^{-1} + x_2^{-1} + x_3^{-1}), & \text{if } J=\{0,2,3\}.
\end{array}
\right.
\end{eqnarray}
Since $x_t$ ($t=1,2,3$) are distinct \(n\)-th roots of unity and $x_1^{-1} + x_2^{-1} + x_3^{-1} = \overline{x_1 + x_2 + x_3}$ is the complex conjugate of $x_1 + x_2 + x_3$, we have $|N_{K/\mathbb{Q}}(x_1 x_2 x_3)| = 1$ and $N_{K/\mathbb{Q}}(x_1^{-1} + x_2^{-1} + x_3^{-1}) = N_{K/\mathbb{Q}}(x_1 + x_2 + x_3)$.
Since $(x_j-x_i)\neq 0$ and \(x_1x_2x_3\neq0\), \(\Delta_S=0\) if and only if \(x_1+x_2+x_3=0\).

If \(x_1+x_2+x_3=0\), there exists an \(n\)-th root of unity \(\omega\) such that
\[
\{x_1,x_2,x_3\} = \{\omega,\;\omega\zeta_3,\;\omega\zeta_3^2\},
\]
where \(\zeta_3 = e^{2\pi i/3}\) is a primitive $3$-rd root of unity. This means that $3\mid n$, which contradicts the assumption that $3\nmid n$. Thus, \(x_1+x_2+x_3 \neq 0\).
Consequently, $\Delta_S \neq 0$ in \(\mathbb{Q}(\zeta_n)\) for any $3$-subset $S = \{l_1,l_2,l_3\}$ of the column index set $\{1,\dots,n\}$.
The desired conclusions follow from definitions, Theorems \ref{non-rs} and \ref{mds2}.
\end{proof}

In particular, from Theorem \ref{non-rs} we construct some non-RS cyclic MDS codes, which are described in Theorem \ref{binary-rs}.

\begin{theorem}\label{binary-rs}
Let the symbols and notations be as above. Let $s\geq 3$ be a positive integer, $q = 2^s$, $n=q+1$ and $p = 2$.
Let $J = \{0, 1, 2, 4, \dots, 2^{k-2}\}$ with $k \ge 4$ and $2^{k-2} < n/2$.
Let $G$ be the $k \times n$ matrix over the cyclotomic field $\mathbb{Q}(\zeta_n)$ given by (\ref{eq-G}). Let $\mathfrak{p}$ be a prime ideal of $\mathbb{Z}[\zeta_n]$ lying above $2$ and let $\rho : \mathbb{Z}[\zeta_n] \to \mathbb{Z}[\zeta_n]/\mathfrak{p}$ be the reduction homomorphism.
Denote by $\overline{G} = \rho(G)$ the matrix obtained by reducing $G$ modulo $\mathfrak{p}$. Then we have the following results.

\begin{enumerate}
\item[(i)] The linear code $\overline{\C}$ generated by $\overline{G}$ is an $[n, k]$ cyclic code over the finite field $\mathbb{F}_{q^2}$. Moreover, $\overline{\C}$ is an  MDS code if and only if no $2m$ distinct $n$-th roots of unity in $\mathbb{F}_{q^2}$ sum to zero for any integer $1 \le m \le \lfloor k/2 \rfloor$. Furthermore, if $\overline{\C}$ is an MDS code, it is not equivalent to an RS code (i.e., of non-RS type).
\item[(ii)] In particular, for any $k \in \{4, 5\}$ and $s \ge 3$ satisfying $2^{k-2} < n/2$, $\overline{\C}$ is an $[n, k, n-k+1]$ cyclic MDS code of non-RS type over $\mathbb{F}_{q^2}$.
\end{enumerate}
\end{theorem}

\begin{proof}
(i). We divide the proof into three parts as follows.
\begin{enumerate}
    \item We first prove that the code $\overline{\C}$ is over the finite field $\mathbb{F}_{q^2}$. Since $n = 2^s + 1$ is odd, we have $2 \nmid n$. By Lemma \ref{lem-CF}, we only need to show that the multiplicative order $f = \ord_n(2)$ of $2$ modulo $n$ is $2s$, i.e., $f=2s$.

Since $n = 2^s + 1$, we have $2^s \equiv -1 \pmod{n}$ and $2^{2s} \equiv 1 \pmod{n}$. Since $f$ is the smallest positive integer such that $2^f \equiv 1 \pmod{n}$, $f\leq 2s$.

By definition, $2^f \equiv 1 \pmod{n}$, which implies $2^f - 1 \ge n = 2^s + 1$, so $2^f > 2^s$ and thus $f > s$. If $s < f < 2s$, we write $f = s + t$ with $1 \le t < s$. Then
\[
1 \equiv 2^f = 2^{s+t} = 2^s \cdot 2^t \equiv -2^t \pmod{n},
\]
so $2^t + 1 \equiv 0 \pmod{n}$, i.e., $n \mid (2^t + 1)$. However, since $1 \le t < s$, we have $0 < 2^t + 1 < 2^s + 1 = n$, which is a contradiction. Therefore, $f = 2s$, and the residue field $\mathbb{Z}[\zeta_n]/\mathfrak{p}$ is isomorphic to $\mathbb{F}_{2^{2s}} = \mathbb{F}_{(2^s)^2} = \mathbb{F}_{q^2}$.

    \item Next, we show that every $k \times k$ minor of $\overline{G}$ is non-zero in $\mathbb{F}_{q^2}$ (i.e., $\overline{\C}$ is an MDS code) if and only if no $2m$ distinct $n$-th roots of unity in $\mathbb{F}_{q^2}$ sum to zero for any $1 \le m \le \lfloor k/2 \rfloor$.

   Let $S = \{l_1, \dots, l_k\}$ be any $k$-subset of distinct column indices in $\{1, \dots, n\}$, and let $\alpha_t = \rho(\zeta_n^{l_t - 1}) \in \mathbb{F}_{q^2}$ for $1 \le t \le k$. Since $2 \nmid n$, the reduction map $\rho$ is injective on the $n$-th roots of unity, so $\alpha_1, \dots, \alpha_k$ are $k$ distinct $n$-th roots of unity in $\mathbb{F}_{q^2}^\times$. The corresponding $k \times k$ submatrix of $\overline{G}$ is
    \[
    \overline{G}_S = \begin{bmatrix}
    1 & 1 & \cdots & 1 \\
    \alpha_1 & \alpha_2 & \cdots & \alpha_k \\
    \alpha_1^2 & \alpha_2^2 & \cdots & \alpha_k^2 \\
    \vdots & \vdots & \ddots & \vdots \\
    \alpha_1^{2^{k-2}} & \alpha_2^{2^{k-2}} & \cdots & \alpha_k^{2^{k-2}}
    \end{bmatrix}.
    \]
    Subtracting the first column from columns $2, \dots, k$ and using the Frobenius identity $(u - v)^{2^r} = u^{2^r} - v^{2^r}$ in characteristic $2$, we expand the determinant along the first row to obtain a $(k-1) \times (k-1)$ determinant \cite[Lemma 3.51]{LidlNiederreiter1997}:
    \begin{equation}\label{eq-moore}
    \det(\overline{G}_S) = \det\Bigl( (\alpha_t - \alpha_1)^{2^{i-2}} \Bigr)_{2 \le i, t \le k} = \prod_{(c_2, \dots, c_k) \in \mathbb{F}_2^{k-1} \setminus \{\mathbf{0}\}} \sum_{t=2}^{k} c_t (\alpha_t - \alpha_1).
    \end{equation}
    In characteristic $2$, setting $c_1 = \sum_{t=2}^{k} c_t \in \mathbb{F}_2$, each linear factor in (\ref{eq-moore}) can be written as $\sum_{t=2}^{k} c_t (\alpha_t - \alpha_1) = \sum_{t=1}^{k} c_t \alpha_t$, where $(c_1, c_2, \dots, c_k) \in \mathbb{F}_2^k \setminus \{\mathbf{0}\}$ satisfies $\sum_{t=1}^{k} c_t = 0$ in $\mathbb{F}_2$ (i.e., it has even Hamming weight $2m$ for some $1 \le m \le \lfloor k/2 \rfloor$). Therefore, $\det(\overline{G}_S) \neq 0$ for all $k$-subsets $S$ if and only if no $2m$ distinct $n$-th roots of unity in $\mathbb{F}_{q^2}$ sum to zero for any $1 \le m \le \lfloor k/2 \rfloor$.

    \item Finally, suppose that $\overline{\C}$ is an MDS code. Since $k \ge 4$, $J = \{0, 1, 2, 4, \dots, 2^{k-2}\}$ is not an arithmetic progression in $\mathbb{Z}$ and $j_k = 2^{k-2} < n/2$, i.e., $j_k \le (n-1)/2$, it follows from Theorem \ref{non-rs} that $\overline{\C}$ is of non-RS type.
\end{enumerate}
Combining the above three parts completes the proof of (i).

(ii). When $k \in \{4, 5\}$, we have $\lfloor k/2 \rfloor = 2$ and  $2m \in \{2, 4\}$. For $2m = 2$, the sum of any two distinct $n$-th roots of unity $u_1, u_2 \in \mathbb{F}_{q^2}^\times$ is $u_1 + u_2 \neq 0$ in characteristic $2$. For $2m = 4$, suppose that there exist four distinct $n$-th roots of unity $u_1, u_2, u_3, u_4 \in \mathbb{F}_{q^2}^\times$ such that $u_1 + u_2 + u_3 + u_4 = 0$. Since $n = q + 1$, every $n$-th root of unity $u \in \mathbb{F}_{q^2}^\times$ satisfies $u^{q+1} = 1$, i.e., $u^q = u^{-1}$. Raising $u_1 + u_2 + u_3 + u_4 = 0$ to the $q$-th power in characteristic $2$ yields
\[
u_1^{-1} + u_2^{-1} + u_3^{-1} + u_4^{-1} = (u_1 + u_2 + u_3 + u_4)^q = 0.
\]
Thus $u_1 + u_2 + u_3 = u_4$ and $u_1^{-1} + u_2^{-1} + u_3^{-1} = u_4^{-1}$. Multiplying the latter by $u_1 u_2 u_3$ gives $u_1 u_2 + u_1 u_3 + u_2 u_3 = u_1 u_2 u_3 u_4^{-1}$. Substituting these into the elementary symmetric polynomial identity over $\mathbb{F}_2$, we obtain
\begin{eqnarray*}
(u_1 + u_2)(u_1 + u_3)(u_2 + u_3) &=& (u_1 + u_2 + u_3)(u_1 u_2 + u_1 u_3 + u_2 u_3) + u_1 u_2 u_3 \\
&=& u_4(u_1 u_2 u_3 u_4^{-1}) + u_1 u_2 u_3 = 0,
\end{eqnarray*}
which implies $u_1 = u_2$, $u_1 = u_3$, or $u_2 = u_3$, contradicting the assumption that $u_1, u_2, u_3, u_4$ are distinct.
Hence the desired conclusions follow from (i) for any $k \in \{4, 5\}$ and $s$ satisfying $2^{k-2} < n/2$ (which holds for $k=4$ when $s \ge 3$ and for $k=5$ when $s \ge 4$). This completes the proof of (ii).
\end{proof}

Note that the dual of a cyclic MDS code (resp.\ a GRS code) is also a cyclic MDS code (resp.\ a GRS code). Then the dual code $\overline{\C}^\perp$ in (ii) of Theorem \ref{binary-rs}
 immediately yields the following result.

\begin{corollary}\label{cor-binary-dual}
 Let the symbols and notations be as in Theorem \ref{binary-rs}. For any $k \in \{4, 5\}$ and $s \ge 3$ satisfying $2^{k-2} < n/2$, the dual code $\overline{\C}^\perp$ of $\overline{\C}$ is an $[n, n-k, k+1]$ non-RS type cyclic MDS code  over $\mathbb{F}_{q^2}$, whose dimension $n - k = 2^s + 1 - k \in \{2^s - 3, 2^s - 4\}$ can be arbitrarily large as $s$ increases.
\end{corollary}

\begin{example}
Some examples on non-RS type cyclic MDS codes obtained by Theorem \ref{binary-rs} are given in Table \ref{tab4}, which are confirmed by Magma programs.
\end{example}

\begin{table}[H]
\begin{center}

\small
\setlength{\tabcolsep}{4pt}

\caption{Some examples on non-RS type cyclic MDS codes in Theorem \ref{binary-rs}. \label{tab4}}
\begin{tabular}{cccccc} \hline
$(s,q,n)$ & $J$   & $\ord_n(2) = 2s$            & $k$      & \makecell{All $k\times k$ minors \\are non-zero over $\mathbb{F}_{q^2}$}   &\makecell{Is $\overline{\C}$ a non-RS cyclic \\ MDS code over $\mathbb{F}_{q^2}$?} \\ \hline

$(3,8,9)$  & $\{0,1,2,4\}$  &  $6$  & $4$        & Yes          & Yes   \\  [2mm]

$(4,16,17)$  & $\{0,1,2,4\}$ &  $8$   & $4$        & Yes          & Yes   \\  [2mm]

$(4,16,17)$  & $\{0,1,2,4,8\}$  &$8$   & $5$        & Yes          & Yes   \\  [2mm]

$(5,32,33)$  & $\{0,1,2,4\}$  &$10$   & $4$        & Yes         & Yes  \\  [2mm]

$(5,32,33)$ & $\{0,1,2,4,8\}$  &$10$  & $5$       & Yes         & Yes  \\  [2mm]
$(5,32,33)$  & $\{0,1,2,4,8,16\}$  &$10$  & $6$      & No         & No   \\  [2mm]

$(6,64,65)$  & $\{0,1,2,4\}$  &$12$  & $4$       & Yes         & Yes   \\  [2mm]

$(6,64,65)$  & $\{0,1,2,4,8\}$  &$12$  & $5$       & Yes         & Yes  \\  [2mm]

$(6,64,65)$ & $\{0,1,2,4,8,16\}$    &$12$ & $6$     & No  & No \\ [2mm]
\hline
\end{tabular}
\end{center}
\end{table}

\section{Concluding remarks}\label{sec-summary}

In this paper, we mainly investigated the construction of cyclic MDS codes via cyclotomic field reduction and presented a systematic method for constructing cyclic MDS codes over finite fields. We transformed the problem of verifying the MDS property over a finite field into a problem of determining non-zero minors in characteristic zero and computing the prime factors of their algebraic norms.
Many cyclic MDS codes over finite fields were obtained by using this method and many non-RS type cyclic MDS codes were produced.
In particular, compared with the construction method in \cite{chen2025}, the results indicate that the method of this paper is more flexible in generating non-RS type cyclic MDS codes over a wide range of finite fields, specifically for infinitely many primes $p\notin P_{bad}$. An interesting direction for future work is to extend this reduction method to other algebraic number fields and to investigate additional families of defining sets $J$ for constructing new infinite families of non-RS cyclic MDS codes.




\end{document}